\documentclass[submission,copyright,creativecommons,noncommercial,noderivs]{eptcs}
\providecommand{\event}{GANDALF 2012} 

\hypersetup{pdfborder={0 0 0}}
\usepackage{enumerate,amsfonts,graphicx,stfloats,amssymb,amsthm, amsmath,color,multicol,yhmath}
\usepackage{tikz}
\usetikzlibrary{calc,matrix,arrows}
\usepackage[T1]{fontenc}

\newtheorem{theorem}{Theorem}
\newtheorem{lemma}[theorem]{Lemma}
\newtheorem{definition}[theorem]{Definition}
\newtheorem{corollary}[theorem]{Corollary}
\newtheorem{proposition}[theorem]{Proposition}

\newcommand{\str}[1]{{\mathfrak{#1}}}
\newcommand{\Z}{{\mathbb Z}}
\newcommand{\N}{{\mathbb N}}

\newcommand{\more}[1]{$!^\text{\textcolor{red}{[justification needed]}}$}

\newcommand{\klasa}{{\cal K}}
\newcommand{\klasap}{{\cal K}_\Phi}
\newcommand{\G}[2]{\Gamma_{#1}^{\,\mathrm{#2}}}
\newcommand{\ram}[1]{deg_{\geq 7}(#1)}

\newcommand{\PSPACE}{{\sc PSpace}}

\newcommand{\NP}{{\sc NP}}
\newcommand{\EXPTIME}{{\sc ExpTime}}
\newcommand{\NEXPTIME}{{\sc NExpTime}}
\newcommand{\APSPACE}{{\sc APSpace}}

\title{Satisfiability vs. Finite Satisfiability in Elementary Modal Logics}
\author{Jakub Michaliszyn, Jan Otop, Piotr Witkowski
\institute{Institute of Computer Science\\
University of Wroc\l{}aw\\
Wroc\l{}aw, Poland\thanks{The first author was supported by Polish Ministry of Science and Higher Education research grant N N206 371339.}}
\email{\{jmi,jotop,pwit\}@cs.uni.wroc.pl}}
\def\titlerunning{Satisfiability vs. Finite Satisfiability in Elementary Modal Logics}
\def\authorrunning{J. Michaliszyn, J. Otop \& P. Witkowski}

\begin{document}
  \maketitle
  \begin{abstract}
We study elementary modal logics, i.e. modal logic considered over first-order definable classes of frames. The classical semantics of modal logic allows infinite structures, but often practical applications require to restrict our attention to finite structures. Many decidability and undecidability results for the elementary modal logics were proved separately for general satisfiability and for finite satisfiability \cite{HS11, KMO11, KM12, lics12}. In this paper, we show that there is a reason why we must deal with both kinds of satisfiability separately --- we prove that there is a universal first-order formula that defines an elementary modal logic with decidable (global) satisfiability problem, but undecidable finite satisfiability problem, and, the other way round, that there is a universal formula that defines an elementary modal logic with decidable finite satisfiability problem, but undecidable general satisfiability problem.
  \end{abstract}

\section{Introduction}
Modern modal logic, founded by Clarence Lewis in 1910, is a formalism that involves the use of the expressions ``possibly'' and ``necessarily''. Formally, modal logic extends propositional logic by some new constructions, of which two most important were $\Diamond \varphi$ and $\square \varphi$,
originally read as \emph{$\varphi$ is possible} and \emph{$\varphi$ is necessary}, respectively. A typical question was,
given a set of axioms $\cal A$, corresponding usually to some intuitively acceptable aspects of truth, 
what is the logic defined by $\cal A$, i.e.~which formulas are provable from $\cal A$ in a Hilbert-style system. 

One of the most important steps in the history of modal logic was the invention in the late 1950s and early 1960s of a formal semantics based on the notion of the so-called Kripke structures. Basically, a Kripke structure is a directed graph, called a \emph{frame}, together with a valuation of propositional variables.  Vertices of this graph are called \emph{worlds}. For each world truth values of all propositional variables can be defined independently. In this semantics, $\Diamond \varphi$ means  \emph{the current world is connected to some world in which $\varphi$ is true}; and $\square \varphi$, equivalent to $\neg \Diamond \neg \varphi$, means \emph{$\varphi$ is true in all worlds to which the current world is connected}.  

It appeared that there is a beautiful connection between syntactic and semantic approaches to modal logic \cite{sahlqvist}: logics defined by axioms can be often equivalently defined by restricting classes of frames. E.g., the axiom $\Diamond \Diamond P \rightarrow \Diamond P$ (\emph{if it is possible that $P$ is possible, then $P$ is possible}), is valid precisely in the class of transitive frames; the axiom $P \rightarrow \Diamond P$ (\emph{if $P$ is true, then $P$ is possible}) -- in the class of reflexive frames, $P \rightarrow \square \Diamond P$ (\emph{if $P$ is true, then it is necessary that $P$ is possible}) -- in the class of symmetric frames, and the axiom $\Diamond P \rightarrow \square \Diamond P$ (\emph{if $P$ is possible, then it is necessary that $P$ is possible}) -- in the class of Euclidean frames. 

Many important classes of frames, in particular all the classes we mentioned above, can be defined by simple first-order formulas. For a given first-order sentence $\Phi$ over the signature consisting of a single binary symbol $R$ we define $\klasap$ to be the set of those frames which satisfy $\Phi$. 

It is not hard to see that some modal logics defined by a first-order formula are undecidable. A stronger result was presented in \cite{HS96}---it was shown that there exists a universal first-order formula with equality such that the global satisfiability problem over the class of frames that satisfy this formula is undecidable. In \cite{HS11}, this result was improved --- it was shown that equality is not necessary. The proof from \cite{HS11} works also for local satisfiability. Finally, in \cite{KMO11} it was shown that even a very simple formula with three variables without equality may lead to undecidability.

Decidability for various classes of frames can be shown by employing the so-called standard translation of modal logic to first-order logic. Indeed, the satisfiability of a modal formula $\varphi$ in $\klasap$ is equivalent to satisfiability of $st(\varphi) \wedge \Phi$, where $st(\varphi)$ is the standard translation of $\varphi$. In this way, we can show that even multimodal logic is decidable in any class defined by two-variable logic \cite{fodwa}, even extended with linear order \cite{Otto} or equivalence closures of two distinguished binary relations \cite{self}. 

A number of decidability results may be obtained by adapting the results for the guarded fragment \cite{gradel}. It has been shown that many interesting extensions of this logic are decidable, including some restricted application of fixed-points \cite{fixed} and transitive closures \cite{m08} in guards. These results often can be extended for the  finite satisfiability problem \cite{barany12, KT07}.
The complexity bounds obtained this way, however, are high --- usually between {\sc ExpTime} and {\sc 2NExpTime}.

The classes of frames we mentioned earlier, i.e. transitive, reflexive, symmetric and Euclidean are decidable.
They can be defined by first-order sentences even if we further restrict
the language to universal Horn formulas. Universal Horn formulas were considered in \cite{HS08},
where a dichotomy result was proved, that the satisfiability problem for modal logic over the class of structures defined by a universal Horn formula (with an arbitrary number of variables) is either in \NP{} or  \PSPACE{}-hard.  The authors of \cite{HS08} conjectured that the problem is decidable in \PSPACE{} for all universal Horn formulas. This conjecture was confirmed in \cite{lics12}.

In case of some universal Horn formulas, decidability of corresponding modal logics
is shown in \cite{lics12} by demonstrating the finite model property, i.e. by proving that every modal formula satisfiable over a class of frame $\klasa$ has also a finite model in $\klasa$. However, it is not always possible, as it is not hard to construct a universal Horn formula formula $\Phi$, such that some modal formulas have only infinite models over $\klasap$, the class of frames satisfying $\Phi$. Assume e.g that $\Phi$ enforces irreflexivity and transitivity, and consider the following modal formula: $\Diamond p \wedge \Box \Diamond p$.

This naturally  leads to the question, whether for any universal Horn formula $\Phi$ the finite satisfiability problem for modal logic over $\klasap$ is decidable. This  question is particularly important, if one considers practical applications, in which the structures (corresponding e.g. to knowledge bases or descriptions of programs) are usually required to be finite.

Decision procedures for the finite satisfiability problem for modal and related logics are very often more complex than procedures for general satisfiability. As argued in \cite{Var97}, the model theoretic reason for the good behavior of modal logics is the tree model property. A standard technique is to unravel an arbitrary model into a (usually infinite) tree. In \cite{lics12} we also apply this idea (at least as a starting point of our constructions, as the obtained unravellings have to be sometimes modified to meet the requirements of the universal Horn formula defining the class of frames). Clearly such an approach, is not sufficient if we are interested only in finite models. However, in \cite{KM12}, the decidability of elementary modal logics defined by universal Horn formulas was extended for the finite satisfiability.

\subsection{Our contribution and related work}
In this paper, we study the following questions: Is every decidable elementary modal logic finitely decidable? Is every finitely decidable elementary modal logic decidable? 

The similar question has been studied in \cite{G06} --- they proved that there is a decidable modal logic such that the finite satisfiability problem for this logic is undecidable. However, the result in \cite{G06} is proved using G\"{o}del--L\"{o}b formula, which is known to define a class of frames that cannot be defined by a first-order formula (see \cite{BdRV01}, Example 3.9). Extensions of modal logics that are undecidable but finitely decidable are also known --- for example, the $A\bar{A}B\bar{B}$ fragment of the Halpern--Shoham logic \cite{MPS10}.

Comparing papers \cite{lics12} and \cite{KM12}, we see that there are elementary modal logics that have the satisfiability and the finite satisfiability problems in different complexity classes. For instance, for $\Phi= \forall x y z v . x R y \wedge y R z \wedge v 
R z \Rightarrow x R v$ the global finite satisfiability is \PSPACE{}-complete, while the general global satisfiability is \EXPTIME{}-complete. However, since we do not know whether \PSPACE{} $\neq$ \EXPTIME{}, we cannot formally conclude that their complexity differs. 

Our contribution consists of the two theorems stated below.

\begin{theorem}\label{thm1}
There is a universal first-order formula $\Gamma_1$ such that global satisfiability problem for the modal logic over the class of frames satisfying $\Gamma_1$ is undecidable, but it becomes decidable if we consider only finite structures.
\end{theorem}

\begin{theorem}\label{thm2}
There is a universal first-order formula $\Gamma_2$ such that global satisfiability problem for the modal logic over the class of frames satisfying $\Gamma_2$ is decidable, but it becomes undecidable if we consider only finite structures.
\end{theorem}


Our undecidability proofs involve the technique presented in \cite{lics12}. To make the presentation easier, in Section \ref{stare} we recall this technique. Then, in subsequent sections, we prove the above theorems.  We conclude with some further discussion and open questions.

\section{Preliminaries}
As we work with both first-order logic and modal logic we help the reader to distinguish them in our notation: we denote first-order formulas with Greek capital letters $\Phi$, $\Psi$, $\Gamma$, and modal formulas with Greek lower-case letters $\varphi$, $\psi$, $\tau$, $\lambda$. 
We assume that the reader is familiar with first-order logic and propositional logic. 
 
Formulas of modal logic are interpreted in  Kripke structures, which are triples of the form $\langle W, R, \pi \rangle$, where $W$ is a set of worlds, $\langle W, R \rangle$ is a directed graph called a \emph{frame}, and $\pi$ is a function that assigns to each world a set of propositional variables which are true at this world. We say that a structure  $\langle W, R, \pi \rangle$ is \emph{based} on the frame $\langle W, R \rangle$. For a given class of frames $\cal K$, we say that a structure is $\cal K$-based if it is based  on some frame from $\cal K$. 
We will use calligraphic letters ${\cal M}, {\cal N}$ to denote frames and Fraktur letters $\str{M}, \str{N}$ to denote structures. Whenever we consider a structure $\str M$, we assume that its frame is $\cal M$ and its universe is $M$ (and the same holds for other letters). 

The semantics of modal logic is defined recursively. A modal formula $\varphi$ is (locally) \emph{satisfied} in a world $w$ of a model $\str M=\langle W, R, \pi \rangle$,
denoted as
${\str  M}, w \models \varphi$ if
\begin{enumerate}[(i)]
\item $\varphi = p$ where  $p$ is a variable and $\varphi \in \pi(w)$,
\item $\varphi = \neg \varphi_1$ and not ${\str M}, w \models \varphi$,
\item $\varphi = \varphi_1 \vee \varphi_2$ and either ${\str  M}, w \models \varphi_1$ or ${\str  M}, w \models \varphi_2$,
\item $\varphi = \varphi_1 \wedge \varphi_2$ and both ${\str  M}, w \models \varphi_1$ and ${\str  M}, w \models \varphi_2$,
\item $\varphi = \Diamond \varphi'$ and there exists a world $v\in W$ such that $(w, v) \in R$ and ${\str  M}, v \models \varphi'$,
\item $\varphi = \square \varphi'$ and for all worlds $v\in W$ such that $(w, v) \in R$ we have ${\str  M}, v \models \varphi'$.
\end{enumerate}

By ${\str  M}, w \models \varphi$  we denote that a modal formula $\varphi$ is (locally) \emph{satisfied} in a world $w$ of a model $\str M$. 
We say that a formula $\varphi$ is \emph{globally} satisfied in ${\str  M} $, denoted as ${\str  M} \models \varphi$, if for all worlds $w$ of ${\str  M}$,  we have ${\str  M}, w \models \varphi$. By $|\varphi|$ denote the length of $\varphi$.

For a given class of frames $\klasa$, we say that a formula $\varphi$ is \emph{locally} (resp.~\emph{globally}) ${\klasa}$-\emph{satisfiable} if there exists a $\klasa$-based structure ${\str M}$, and a world $w \in W$ such that ${\str M}, w \models \varphi$ 
(resp.~${\str M} \models \varphi$).

For a given formula $\varphi$, a Kripke structure $\str M$, and a world $w \in W$ we define the \emph{type} of $w$ (with respect to $\varphi$) in $\str M$ as $tp_{\str  M}^{\varphi}(w) = \{ \psi : {\str  M}, w \models \psi $ and $\psi$ is subformula of $\varphi \}$. We write $tp_{\str  M}(w)$ if the formula is clear from context. Note that $|tp_{{\str  M}}^{\varphi}(w)| \leq |\varphi|$. 

The class of universal first order sentences is defined as a subclass of first--order sentence such that each sentence is of the form $\forall \vec{x} \,\Psi(\vec{x})$, where $\Psi(\vec{x})$ is quantifier--free formula, and the only allowed relational symbols are $R$ (interpreted as the successor relation in modal logic) and $=$ (interpreted as usual). As we work only with universally quantified formulas, we often skip the quantifier prefix, e.g., write $xRx$ instead of $\forall x . xRx$.

We define the \emph{local} (resp.~\emph{global}) \emph{satisfiability problem} $\klasa$-SAT 
(resp.~global ${\klasa}$-SAT) as follows. For a given modal formula, is this formula locally (resp.~globally) ${\klasa}$-satisfiable? 
For a given first-order formula $\Phi$, we define  $\klasa_{\Phi}$ as the class of frames satisfying $\Phi$. We will be interested
in local and global $\klasa_{\Phi}$-SAT problems. The \emph{finite} (global) satisfiability problem, (global) $\klasa$-FINSAT, is defined in the same way, but we are only interested in finite structures.




Our undecidability proofs are based on the following observation. In universal first--order logic, we cannot express that some world has a large out-degree. However, we are able to define some properties of such worlds. To this end, we define an abbreviation $deg_{\geq m}(v)$ that uses the fresh variables $u^v_1, \dots, u^v_m$ as follows.

\[ deg_{\geq m}(v)=\bigwedge_{1 \leq i \leq m} (v R u^v_i) \wedge \bigwedge_{1 \leq i < j \leq m} \neg u^v_i = u^v_j \]

We say that a world is \emph{ramified} if its out-degree is at least $7$. For example, the formula $deg_{\geq 7}(v) \Rightarrow vRv$ says that all ramified worlds are reflexive. Using universally quantified first order formulas we will define transitive
relation on ramified worlds. We will then be able to distinguish finite structures from infinite ones, because any such order
forces some world in a finite structure to be reflexive and we can define an order in a non-finite structure to be irreflexive.

\subsection{Domino systems}
By $\Z_m$ we denote the set $\{0, 1, \dots, m-1\}$.

\begin{definition} 
 A \emph{domino system} is a tuple ${\cal D}=(D, D_H, D_V)$, where $D$ is a set
of domino pieces and $D_H, D_V \subseteq D \times D$ are binary
relations specifying admissible horizontal and vertical adjacencies.
We say that $\cal D$ \emph{tiles} $\N \times \N$  if
there exists a function $t:\N \times \N \rightarrow {\cal D}$ 
such that $\forall i,j \in \N$ we have $(t(i,j),t(i+ 1,j)) \in D_H$ and $(t(i,j),t(i,j+1)) \in D_V$.
Similarly, $\cal D$ tiles $\Z_m \times \Z_l$, for $m,l \in \N$, if there exists 
$t:\Z_m \times \Z_l \rightarrow {\cal D}$ such that 
$(t(i,j),t(i+1 {\scriptstyle\mod m},j)) \in D_H$ and $(t(i,j),t(i,j + 1 {\scriptstyle\mod l})) \in D_V$.
\end{definition}

The following theorem comes from \cite{Berger66, GK72} (see also \cite{BorgerGG1997} for more modern
exposition).

\begin{theorem} The following problems are undecidable: \label{doomino}
\begin{enumerate}[(i)]
\item For a given domino system $\cal D$ determine if  $\cal D$ tiles $\N \times \N$.
\item For a given domino system $\cal D$ determine if there exists $m \in \N$ such that $\cal D$ tiles $\Z_m \times \Z_m$.
\end{enumerate}
\end{theorem}

\section{Logic with undecidable SAT and FINSAT}\label{stare}

As a warm-up exercise, we recall the undecidability result from \cite{lics12}, as it will be a part of our later proofs. 
We define $\Gamma$ as
$$x R y \wedge x R u \wedge u R z \wedge deg_{\geq 2}(x) \wedge deg_{\geq 4}(u) \wedge deg_{\geq 2}(z) \Rightarrow y R z$$

\begin{theorem}\label{undie}
The global satisfiability problem and the finite global satisfiability problem for modal logic over ${\cal K}_{\Gamma}$ are undecidable.
\end{theorem}

We work with signatures consisting of a single binary symbol $R$, and a number of unary symbols, including $P_{ij}$, $A_{ij}$, and $E_{ij}^k$ for $i,j, k \in [0,2]$. Structures over such signatures can be naturally viewed as Kripke structures in which $R$ is the accessibility relation, and unary relations describe valuations of propositional variables. 

The structure $\str{G}_\N$ illustrated in Fig.~\ref{grid} is a model of $\Gamma$. The idea of the proof is similar to the proof of the undecidability presented in \cite{KMO11}. The main difficulty here is to enforce $\str{G}_\N$ to be a grid-like structure. This is obtained using $\Gamma$ and a modal formula, which enforces some worlds (labeled $A_{ij}$ in Fig.~\ref{grid}) to have appropiate out-degree.

\begin{figure*}
\usetikzlibrary{calc,matrix,arrows}
\tikzstyle{vertex}=[circle,draw=black,minimum size=20pt,inner sep=0pt]
\tikzstyle{dots}=[minimum size=15pt,inner sep=0pt]
\tikzstyle{dvertex}=[circle,draw=black,minimum size=15pt, inner sep=0pt]

\newcounter{orggri@a}
\newcounter{orggri@b}
\newcounter{orggri@c}

\begin{tikzpicture}[shorten >=0.5pt,->, scale=3.1]

\foreach \xx in {0,1,2}
\foreach \yy in {0,1}
{
	\setcounter{orggri@a}{\xx}
	\setcounter{orggri@b}{\xx}
	\setcounter{orggri@c}{\yy}
	\addtocounter{orggri@a}{\yy} 
	\addtocounter{orggri@b}{1}
	\addtocounter{orggri@c}{1}
	\setcounter{orggri@a}{\intcalcMod{\value{orggri@a}}{3}}
	\setcounter{orggri@b}{\intcalcMod{\value{orggri@b}}{3}}
	\setcounter{orggri@c}{\intcalcMod{\value{orggri@c}}{3}}
	\def \pos {p\arabic{orggri@a}}
	\def \xxx {\arabic{orggri@b}}
	\def \yyy {\arabic{orggri@c}}

 	\foreach \var/\x/\y/\name in 
	{
		P_{\xx\yy}/0/0/1\pos, P_{\xx\yyy}/0/1/2\pos,
		P_{\xxx\yy}/1/0/3\pos, P_{\xxx\yyy}/1/1/4\pos,
	A_{\xx\yy}/0.5/0.5/5\pos}
    \node[vertex] (G-\name) at (\x + \xx,\y +\yy) {$\var$};

	\foreach \var/\x/\y/\name in 
	{
		E_{\xx\yy}^0/0.25/0.75/8\pos,
		E_{\xx\yy}^1/0.65/0.18/9\pos,
		E_{\xx\yy}^2/0.80/0.35/7\pos}
   \node[dvertex] (G-\name) at (\x + \xx,\y+\yy) {${}_{\var}$}; 

	\foreach \from/\to in {
		1\pos/2\pos,1\pos/3\pos,2\pos/4\pos,
		3\pos/4\pos,1\pos/5\pos}
    \draw[arrows=-angle 90] (G-\from) -> (G-\to);

	\foreach \from/\to in {5\pos/8\pos, 5\pos/9\pos}
    \draw[arrows={angle 90-angle 90}] (G-\from) -> (G-\to);

	\foreach \from/\to in {5\pos/7\pos}
	\draw[arrows={angle 90-angle 90}]  
   (G-\from) edge (G-\to);

	\foreach \from/\to in {5\pos/4\pos}
	\draw[arrows={-angle 90}]  (G-\from) edge  (G-\to);
}

\def \pos {ap}
\def \xx {3.8}
\def \yy {0.5}

\foreach \var/\x/\y/\name in 
	{
	a/0/0/1ap, 
	a_{y}/0/1/2ap,
	a_{x}/1/0/3ap, 
	a_{z}/1/1/4ap,
	a_{u}/0.5/0.5/5ap}
    \node[vertex] (G-\name) at (\x + \xx,\y + \yy) {$\var$};

	\foreach \var/\x/\y/\name in 
	{
		u_1^{a_u}/0.25/0.75/8\pos,
		u_2^{a_u}/0.65/0.18/9\pos,
		u_3^{a_u}/0.80/0.35/7\pos}
   \node[dvertex] (G-\name) at (\x + \xx,\y+\yy) {${}_{\var}$};

\foreach \from/\to in {
		1\pos/2\pos,1\pos/3\pos,2\pos/4\pos,
		3\pos/4\pos,1\pos/5\pos}
    \draw[arrows=-angle 90] (G-\from) -> (G-\to);

\foreach \from/\to in {5\pos/8\pos, 5\pos/9\pos}
    \draw[arrows={angle 90-angle 90}] (G-\from) -> (G-\to);

\foreach \from/\to in {5\pos/7\pos}
	\draw[arrows={angle 90-angle 90}]  
   (G-\from) edge (G-\to);

\foreach \from/\to in {5\pos/4\pos}
	\draw[arrows={-angle 90}]  (G-\from) edge  (G-\to);

\coordinate [label=center:$\vdots$] (A) at (1.5,2.2);

\coordinate [label=center:$\cdots$] (A) at (3.3, 1);

\end{tikzpicture}
\centering
\caption{The structure ${\str{G}}_\N$ (left) and its fragment (right). The universe of ${\str{G}}_\N$ is $\N \times \N$. }\label{grid}
\end{figure*}
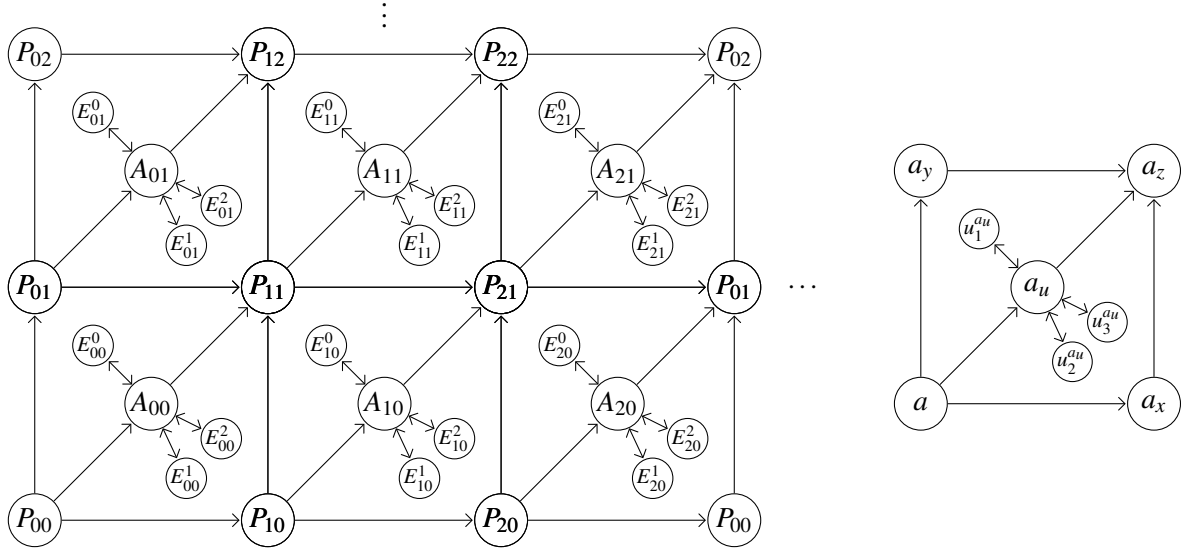

To get the undecidability we construct a modal formula $\tau$ such that any ${\cal K}_\Gamma$-based model $\str{M} \models \tau$ locally looks like a grid. To this end, we use the following definition. We say that a predicate $P$ is \emph{followed by} $Q_1$, $Q_2$, \dots, $Q_m$ in some model if each world of this model satisfying $P$ has, for all $i \in [1, m]$, a successor satisfying $Q_i$, and all successors of of this world satisfy $Q_1 \vee Q_2 \vee \dots Q_m$. This property can be expressed by a modal formula $P \Rightarrow (\bigwedge_{i \in [1, m]} \Diamond Q_i \wedge \square \bigvee_{i \in [1, m]} Q_i)$.

The formula $\tau$ says that:
\begin{enumerate}[(i)]
\item\label{tau1} each element is labeled with exactly one of predicates from the set $\{P_{ij} | i, j \in [0,2]\}$ $\cup \{A_{ij} | i, j \in [0,2]\}$  $\cup \{E_{ij}^k | i,j,k \in [0,2]\}$;
\item\label{tau3} $P_{ij}$ is followed by $P_{(i+1 \mod{3}) j}$, $P_{i(j+1 \mod{3})}$, $A_{ij}$  for all $i,j \in [0,2]$;
\item\label{tau4} $A_{ij}$ is followed by $P_{(i+1\mod{3})(j+1\mod{3})}$, $E_{ij}^0$, $E_{ij}^1$, $E_{ij}^2$  for all $i,j \in [0,2]$;
\item\label{tau5} $E_{ij}^k$ is followed by $A_{ij}$  for all $i,j,k \in [0,2]$.
\end{enumerate}

All these properties are easy to express in modal logic. Observe that each non-empty model of this formula contains a world satisfying $P_{00}$. If we consider now any world $a$ satisfying, e.g., $P_{00}$ in a model, we see that the property (\ref{tau3}) of $\tau$ enforces the existence of its horizontal successor $a_x$ satisfying $P_{10}$, its vertical  successor $a_{y}$ satisfying $P_{01}$ and its upper-right diagonal successor $a_u$ satisfying $A_{00}$ (see Fig.~\ref{grid}). By (\ref{tau4}), the element satisfying $A_{00}$ has four successors, including $a_z$ satisfying $P_{11}$. 
It should be clear that when we instantiate $\Gamma$ with the worlds $a$, $a_x$ (or $a_{y}$), $a_u$, and $a_z$, then the antecedent of the $\Gamma$ is satisfied, and that implies the edges from $a_x$ and $a_y$ to $a_z$ (see Fig.~\ref{grid}). Therefore, there is a homomorphism from $\str{G}_\N$ to any model of $\Gamma$.

Now, for a given domino system ${\cal D} = (D, D_H, D_V)$ we define $$\lambda^{\cal D}
=\lambda_0 \wedge \bigwedge\limits_{0 \le i,j \le 2} (\lambda_{ij}^H \wedge \lambda_{ij}^V).$$ 
For every $d \in D$ we use a fresh propositional letter $P_d$. Formula  
$\lambda_0$ says that each world satisfying some $P_{ij}$ contains a domino piece, $\lambda^H_{ij}$ and $\lambda^V_{ij}$
say that pairs of elements satisfying horizontal and vertical adjacency relations respect $D_H$ and $D_V$, respectively.
$$\lambda^H_{ij}=\bigwedge_{d \in D} ((P_d  \wedge P_{ij}) \rightarrow \square
 (P_{(i+1 \mod 3)j} \rightarrow \hspace*{-10pt} \bigvee_{d': (d,d') \in D_H} \hspace*{-10pt} P_{d'})), $$
$$\lambda^V_{ij}=\bigwedge_{d \in D} ((P_d  \wedge P_{ij}) \rightarrow \square
 (P_{i(j+1 \mod 3)} \rightarrow \hspace*{-10pt} \bigvee_{d': (d,d') \in D_V} \hspace*{-10pt} P_{d'})). $$

The undecidability follows from the following fact
\newtheorem{fact}[theorem]{Fact}
\begin{fact}
$\cal D$ tiles $\N \times \N$ if and only if there exists a $\klasa_\Gamma$-based model of $\tau \wedge \lambda^{\cal D}$.
\end{fact}
\begin{proof}
Proof of ``only if'' part.
Let $t$ be a tiling of $\N \times \N$. We construct $\str{M'}$ by extending the labeling of $\str{G}_\N$ in such a way that for every $i,j \in \N$ the element $a_{i,j}$ satisfies $P_{t(i,j)}$. It is readily checked that $\str{M'}$ is as required.

Proof of ``if'' part.
 Let $\str{M}$ be a $\klasa_\Gamma$-based model of $\tau \wedge \lambda^{\cal D}$ and $f$  be a homomorphism from $\str{G}_\N$ into $\str{M}$. We define a tiling $t: \N \times \N \to D$ by setting $t(i,j)=d$ for such $d$ that $f(i,j)$ satisfies $P_d$ (there is at least one such $d$ owing to $\tau$). The formulas $\lambda^H_{ij}, \lambda^V_{ij}$ imply that $t$ is a correct tiling. 
\end{proof}
Finite case is similar --- we show that $\cal D$ tiles $\Z_m \times \Z_m$ for some $m$ if and only if there exists a finite $\klasa_\Gamma$-based model of $\tau \wedge \lambda^{\cal D}$ (see \cite{lics12}).

\section{Undecidable logic that is finitely decidable}
In this section, we prove Theorem \ref{thm1}. But before we define $\Gamma_1$, we provide some intuitions.

Remark that the intended model ${\str{G}}_\N$ of $\Gamma$ from Section \ref{stare} does not contain directed cycles with more that two worlds. Let us present the main idea using bimodal logic, i.e. we use also modalities $\Diamond ', \square '$ that are interpreted as a relation $R'$ in frame. We define $\Gamma_{bi} = \Gamma \wedge \G{bi}{oc}$, where $\G{bi}{oc}$ says that $R'$ is a strict linear order that contains
non-symmetric edges from $R$, i.e. $xRy \wedge \neg yRx \Rightarrow xR'y$. It is not hard to see that we can add $R'$ edges to the model $\str G_\N$ from Fig. \ref{grid} in a way such that the resulting model satisfy $\Gamma_{bi}$. 

Consider any finite model of $\Gamma_{bi}$ and some modal formula $\varphi$. The linear order defined by $R'$ has to have the greatest world, namely $v$. Since $v$ has no successors and satisfies $\varphi$ (recall that we consider the global satisfiability problem), the structure that contains only $v$ is also model of $\varphi$ and, as a substructure, of $\Gamma_{bi}$. Therefore, the logic has a singe-world model property and the finite satisfiability is decidable in \NP{}.

We cannot directly exploit this idea in unimodal case, making the only relation $R$, because it would imply that all elements of grid are connected. What we are going to do is to make $R$ transitive on the substructure that contains all ramified worlds.

\subsection{The formula}
We show a modification of the formula $\Gamma$ from Section \ref{stare}, such that 
the global satisfiability remains undecidable, but the finite global satisfiability becomes decidable. Define $\Gamma_1$ as \[ \G{1}{two} \Rightarrow \G{1}{irr} \wedge \G{1}{trans} \wedge \G{1}{grid}\] where \bigskip

\noindent \begin{tabular}{l c l}
 $\G{1}{two}$ &$=$& $\ram{v_1} \wedge \ram{v_2} \wedge \neg v_1 = v_2$ \\
 $\G{1}{irr}$ &$=$& $\neg xRx$ \\
 $\G{1}{trans}$ &$=$& $\ram{x}  \wedge x R y \wedge y R z  \Rightarrow x R z$ \\
 $\G{1}{grid}$ &$=$& $x R y \wedge x R u \wedge u R z \wedge deg_{\geq 2}(x) \wedge deg_{\geq 4}(u) \wedge deg_{\geq 2}(z) \Rightarrow y R z \vee x R z$
\end{tabular}
\medskip

In the following subsections we show that $\klasa_{\Gamma_1}$-SAT is undecidable, but $\klasa_{\Gamma_1}$-FINSAT is decidable. This clearly leads to the proof of Theorem \ref{thm1}. 

\subsection{The decidability}
\begin{lemma}
\label{thm1dec}
$\klasa_{\Gamma_1}$-FINSAT is decidable.
\end{lemma}

\begin{proof}
We are going to show that if a modal formula has a finite model $\str M$ satisfying $\Gamma_1$, 
then it has a model containing at most one ramified world. In such a structure the formula $\G{1}{two}$ 
is false, so $\Gamma$ is trivially satisfied and the only problem will be to ensure that the structure
satisfies a modal formula.  

Assume that $\str M$ is a finite model of $\Phi$ and $\varphi$ and contains more than one ramified world.
Then, there is a ramified world that has no ramified descendant in $\str M$. Indeed, $\G{1}{trans}$ guarantees that $R$ is transitive on ramified worlds, and, by $\G{1}{irr}$, there are no cycles. By finiteness of $\str M$, there has to be a ramified world $v$ without ramified successors. Again, by $\G{1}{trans}$, it has no ramified descendants.
A substructure consisting of all worlds reachable from $v$ in $\str M$ is a model of $\Gamma_1$, $\varphi$ and it contains only one ramified world.

It remains to give an algorithm deciding the existence of a finite model of $\varphi$ and $\Gamma_1$ which has at most one ramified world.
Notice that any structure which contains at most one ramified world satisfies $\Gamma_1$, hence 
it suffices to decide the existence of a finite model of $\varphi$ which has at most one ramified world.

We are going to show that a modal formula $\varphi$ has a finite model with at most one ramified world if and only if it has a model based of a (possibly infinite) tree (not necessarily satisfying $\Gamma_1$) such that all ramified worlds in the model have the same type.

If $\varphi$ has a model with at most one ramified world, then the result of standard unraveling procedure~(see \cite{BdRV01}) applied to this model is a model based on some tree such that all ramified worlds have the same type.

Suppose that there is a model $\str M$ of $\varphi$ such that all ramified worlds in $\str M$ have the same type.  
 We define the relation $\approx$ on $M \times M$ as 
 $w_1 \approx w_2$ if and only if worlds $w_1,w_2$ have the same type in $\str M$. Let $f$ be a selection function on equivalence classes of $\approx$.We define $\str N$ as follows.
Let $N=\{ [u]_{\approx} : u \in \str M\}$ be the universe of $\str N$ and 
for all $v_1, v_2 \in N$, we set $v_1 R v_2$ in $\str N$ if and only if $f(v_1)$ has in $\str M$ a successor of type of $f(v_2)$.
Finally, the variable assignment of $v \in \str N$ is consistent the variable assignment of $f(v)$ in $\str M$. 

Clearly, $\str N$ is a finite structure.
One can show by induction that for all $v \in N, w\in M$, $w \in v$ implies $tp_{\str N}(v) = tp_{\str M}(w)$.
In particular, $\str N$ is a model of $\varphi$.
Finally, for every $v \in \str N$, the out-degree of $v$ is not greater than the out-degree of $f(v)$ in $\str M$. 
Since each ramified world in $\str M$ has the same type, there is at most one ramified world in $\str N$.
Thus, we have reduced the problem of existence of finite model of $\varphi$ to the problem of
existence of model of $\varphi$ based on tree whose all ramified worlds have the same type. We shall show an alternating algorithm working in  polynomial space deciding the latter problem. 

The algorithm constructs a model which is tree. First, the algorithm guesses a type $t$ for all ramified worlds in a model. 
Then, it guesses the type of the root and the types of its successors. Next, it
universally branches and guesses the types of successors of a given successors and so on. 
If the current node has more than $6$ successors and its type is different than $t$, it rejects. Clearly, the algorithm works in alternating polynomial space and 
it has an accepting run if and only if there is a model of $\varphi$ based on tree whose all ramified worlds have the same type.

\end{proof}

\subsection{The undecidability}
\begin{lemma}
\label{thm1und}
$\klasa_{\Gamma_1}$-SAT is undecidable.
\end{lemma}
\begin{proof}

In order to show undecidability of the global satisfiability problem, we show that all models of  $\Gamma_1$ and 
a modal formula $\tau_1$ locally look like a grid. The formula $\tau_1$ is based on $\tau$ from Section \ref{stare}. It enforces 
that each world satisfies exactly one of the following unary predicates:
$P_{ij}, A_{ij}, B_{ij}, E_{ij}^k, F_{ij}, G_{ij}^l$ for  $i,j,k \in [0,2]$, $l \in [0,5]$.
The formula $\tau_1$ is a conjunction of the following formulas (note that conditions (\ref{tau3}) and (\ref{tau5}) are the same as in the definition of $\tau$ in  Section \ref{stare}):

\begin{enumerate}[(i)]
\item\label{tau11} each world satisfies exactly one of the unary predicates
 $P_{ij}, A_{ij}, B_{ij}, E_{ij}^k, F_{ij},G_{ij}^l,$ for $i,j,k\in [0,2], l \in [0,5]$;
\item\label{tau13} $P_{ij}$ is followed by $P_{(i+1 \mod{3}) j}$, $P_{i(j+1 \mod{3})}$, $A_{ij}$ for all $i,j \in [0,2]$;
\item\label{tau14} $A_{ij}$ is followed by $P_{(i+1\mod{3})(j+1\mod{3})}$, $F_{ij}$, $E_{ij}^0$, $E_{ij}^1$, $E_{ij}^2$  for all $i,j \in [0,2]$;
\item\label{tau17} $E_{ij}^k$ is followed by $A_{ij}$  for all $i,j,k \in [0,2]$;
\item\label{tau15} $F_{ij}$ and $G_{ij}^k$ are followed by $B_{ij}$  for all $i,j,k \in [0,2]$;
\item\label{tau16} for all $i,j \in [0,2]$, every element satisfying $B_{ij}$ has a successor satisfying $P_{(i+1 \mod 3)(j+1 \mod 3)}$
and six successors satisfying $G_{ij}^0, G_{ij}^1,\ldots G_{ij}^5$ respectively (but can also have
other successors satisfying other propositional variables).
\end{enumerate}

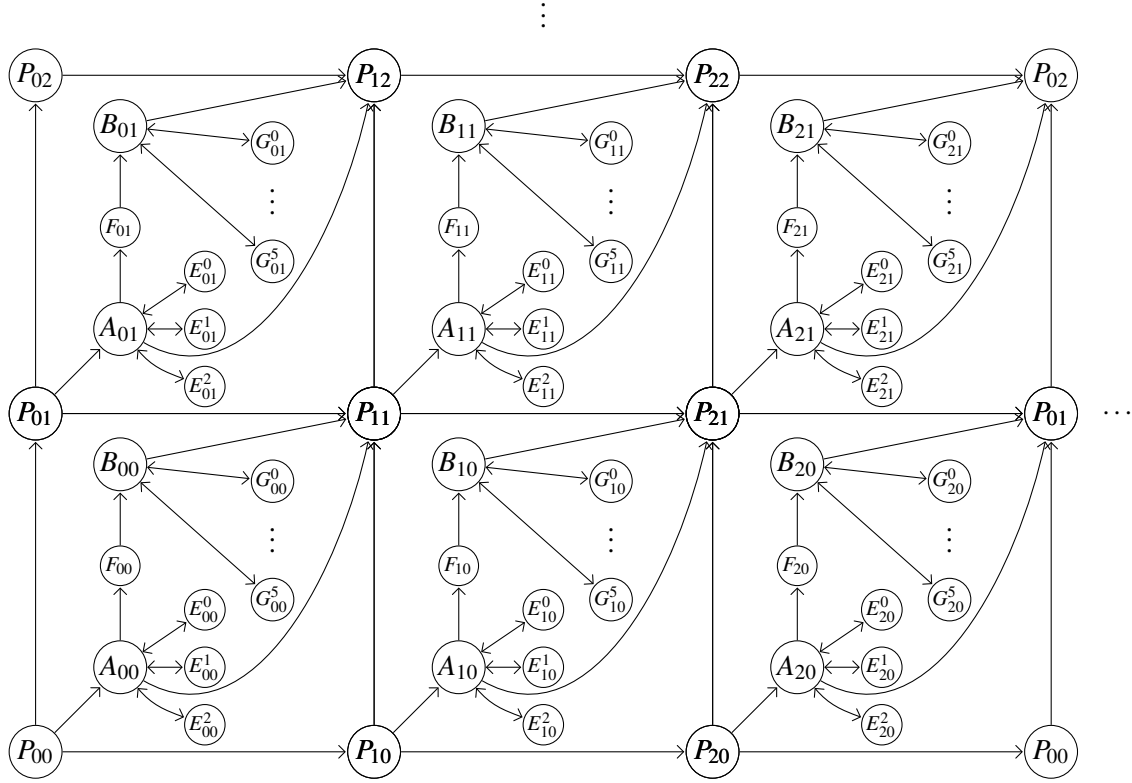
\begin{figure*}[t!]
\centering
\usetikzlibrary{calc,matrix,arrows}

\newcounter{posit@a}
\newcounter{posit@b}
\newcounter{posit@c}

\begin{tikzpicture}[shorten >=0.5pt,->, scale=4.5]

\foreach \xx in {0,1,2}
\foreach \yy in {0,1}
{
	\setcounter{posit@a}{\xx}
	\setcounter{posit@b}{\xx}
	\setcounter{posit@c}{\yy}

	\addtocounter{posit@a}{\yy} 
	\addtocounter{posit@b}{1}
	\addtocounter{posit@c}{1}

	\setcounter{posit@a}{\intcalcMod{\value{posit@a}}{3}}
	\setcounter{posit@b}{\intcalcMod{\value{posit@b}}{3}}
	\setcounter{posit@c}{\intcalcMod{\value{posit@c}}{3}}

	\def \pos {p\arabic{posit@a}}
	\def \xxx {\arabic{posit@b}}
	\def \yyy {\arabic{posit@c}}

 	\tikzstyle{vertex}=[circle,draw=black,minimum size=20pt,inner sep=0pt]
	\tikzstyle{dots}=[minimum size=15pt,inner sep=0pt]
	\tikzstyle{dvertex}=[circle,draw=black,minimum size=15pt, inner sep=0pt]

 	\foreach \var/\x/\y/\name in 
	{
		P_{\xx\yy}/0/0/1\pos,P_{\xx\yyy}/0/1/2\pos,
		P_{\xxx\yy}/1/0/3\pos,P_{\xxx\yyy}/1/1/4\pos,
		A_{\xx\yy}/0.25/0.25/5\pos,B_{\xx\yy}/0.25/0.85/6\pos}
    \node[vertex] (G-\name) at (\x + \xx,\y +\yy) {$\var$};

	\node[dots] (XX) at (0.7+\xx, 0.65+\yy) {\vdots};

	 \foreach \var/\x/\y/\name in 
	 {
		E_{\xx\yy}^0/0.5/0.42/8\pos, E_{\xx\yy}^1/0.5/0.25/9\pos,
		E_{\xx\yy}^2/0.5/0.08/7\pos, 
		F_{\xx\yy}/0.25/0.55/10\pos,
		G_{\xx\yy}^5/0.7/0.45/11\pos, G_{\xx\yy}^0/0.7/0.8/12\pos}
	\node[dvertex] (G-\name) at (\x + \xx,\y+\yy) {${}_{\var}$}; 

	\foreach \from/\to in {
		1\pos/2\pos,1\pos/3\pos,2\pos/4\pos,3\pos/4\pos,
		1\pos/5\pos,6\pos/4\pos,5\pos/10\pos,10\pos/6\pos}
    \draw[arrows=-angle 90] (G-\from) -> (G-\to);

	\foreach \from/\to in {
		5\pos/8\pos, 5\pos/9\pos,6\pos/11\pos,6\pos/12\pos}
    \draw[arrows={angle 90-angle 90}] (G-\from) -> (G-\to);

	\foreach \from/\to in {5\pos/7\pos}
	\draw[arrows={angle 90-angle 90}]  (G-\from) edge [in=160,out=310] (G-\to);

	\foreach \from/\to in {5\pos/4\pos}
	\draw[arrows={-angle 90}]  (G-\from) edge [in=258,out=330] (G-\to);
}

\coordinate [label=center:$\vdots$] (A) at (1.5,2.2);

\coordinate [label=center:$\cdots$] (A) at (3.2, 1);

\end{tikzpicture}
\centering
\caption{The structure ${\str{G}}_1$. Its universe is $\N \times \N$. Some edges from worlds satisfying $B_{ij}$ that follow from $\G{1}{trans}$ are omitted for better readability.}\label{f:extended-cell}
\end{figure*}

First, let us observe that an extension of the standard infinite directed grid is a model of $\Gamma_1$ and $\tau_1$.
Let $\str G_1$ be a grid composed of cells as in Figure \ref{f:extended-cell}. 
Additionally, in each cell, the world labeled by $B_{ij}$ is connected to all worlds reachable from the world 
in the same cell labeled by $P_{(i+1\mod{3})(j+1\mod{3})}$. It is easy to check that $\str G_1$ globally satisfies $\tau_1$.

We will check that $\str G_1$ satisfies $\G{1}{irr} \wedge \G{1}{trans} \wedge \G{1}{grid}$, the consequent of $\Gamma_1$.
Clearly, no world in $\str G_1$ is reflexive. Since for every ramified world, its successors are precisely all words reachable from a particular world, $\str G_1$ satisfies $\G{1}{trans}$.
Finally, the structure $\str G_1$ satisfies $\G{1}{grid}$.
There are two kinds of worlds in $\str G_1$ satisfying $deg_{\geq 4}(u)$, the worlds satisfying $A_{ij}$
and those satisfying $B_{ij}$. 
It is easy to check that if $u$ is a world satisfying $A_{ij}$, then
$x,y,z$ are from the same cell and $\Gamma$ is satisfied.
The worlds satisfying $B_{ij}$ have in $\str G_1$ only predecessors satisfying
$F_{ij}$, for which  $deg_{\geq 2}(x)$ fail, and those that satisfy $B_{i'j'}$. 
However, if $u$ is labeled by $B_{ij}$, it predecessor $x$ is labeled by $B_{i'j'}$, then for every successor $z$  of $u$
we have $x R z$ due to $\G{1}{trans}$.
Hence, $\G{1}{grid}$ is satisfied. Thus $\Gamma_1$ is satisfied.

Second, observe that every model of $\Gamma_1$ and $\tau_1$ contains a model of $\Gamma$ and $\tau$
as a substructure. Let $\str M$ be a model of $\Gamma_1$ and $\tau_1$  and let $\str N$ be a substructure of
$\str M$ resulting from removing worlds labeled by $B_{ij}$, $F_{ij}$ or $G_{ij}^k$. 
Since $\str M$ satisfies $\tau_1$, each world of $\str N$, except worlds satisfying $A_{ij}$,
satisfies $\tau$ and the witnesses are also in $\str N$. 
The worlds satisfying $A_{ij}$  have additional successors in $\str M$
satisfying $F_{ij}$ which are not allowed by $\tau$, but those successors do not belong to $\str N$. 
The worlds satisfying $A_{ij}$ have in $\str N$ successors satisfying $E_{ij}^k$ and 
$P_{(i+1\mod{3})(j+1\mod{3})}$. Hence, $\str N$ globally satisfies $\tau$.

The formula $\tau_1$ implies that $\str M$ contains worlds labeled by $B_{00}, B_{10}, \ldots$, which are different worlds. 
Also, it forces each world labeled by $B_{ij}$ to have seven different successors. 
Hence, for $v_1, v_2$ satisfying $B_{00}, B_{01}$ the formula $\G{1}{two}$ is satisfied.
We see that $\str M$ satisfies $\G{1}{grid}$ and, since $\G{1}{grid}$ is a universal formula and $\str N$ is a substructure of $\str M$, $\str N$ satisfies $\G{1}{grid}$.

Now, by examining all possible labeling of worlds by 
$P_{ij}, A_{ij}, E_{ij}^k$ with $i,j,k \in [0,2]$, consistent with $\tau$,
we observe that $\str N$ satisfies:
\begin{equation}
 x R u \wedge u R z \Rightarrow \neg x R z
\label{eq:anti-trans}
\end{equation}
Clearly, $\G{1}{grid}$ and \eqref{eq:anti-trans} imply $\Gamma$. Thus, $\str N$ satisfies $\Gamma$.

Hence, we have shown that each model of $\Gamma_1$ and $\tau_1$ is a grid-like structure. Thus, we may proceed as in Section \ref{stare} and employ $\lambda_{\cal D}$
to show undecidability of the global $\klasa_{\Gamma_1}$-SAT.
\end{proof}

\section{Decidable logic that is finitely undecidable}
Now we prove
Theorem \ref{thm2}, i.e. we show a formula $\Gamma_2$ 
such that the global satisfiability of modal logic over $\klasa_{\Gamma_2}$ is decidable, but its finite global 
satisfiability problem is undecidable. 

The idea follows from the tree-model property for modal logic \cite{Var97}. Observe that when we consider the finite global satisfiability problem, then modal logic lack the tree-model property. For instance, consider a formula $\Diamond \top$, and its finite model consisting of $n$ worlds. Since every world has a successor, there are at least $n$ edges in the model, and therefore it cannot be a tree.

The first-order formula that we are going to define will allow us to force the existence of a reflexive world, and the same formula will provide additional constraints for structures that contain reflexive worlds. On the other hand, in the general satisfiability case we will prove a property similar to the tree-model property and use it to prove the decidability.

\subsection{The formula}
We define $\Gamma_2$ as 
$(\G{2}{refram}\Rightarrow\G{2}{grid}) \wedge \G{2}{trans}$
where
\medskip

\noindent\begin{tabular}{l c l}
 $\G{2}{refram}$ &$=$&  $deg_{\geq 7}(v) \wedge v R v$\\
 $\G{2}{trans}$ &$=$& $\ram{x}  \wedge x R y \wedge y R z  \Rightarrow x R z$ \\
 $\G{2}{grid}$ &$=$& $x R y \wedge x R u \wedge u R z \wedge deg_{\geq 2}(x) \wedge deg_{\geq 4}(u) \wedge deg_{\geq 2}(z) \Rightarrow y R z \vee x R z$
\end{tabular}
\medskip

Note that $\G{2}{trans} = \G{1}{trans}$ and $\G{2}{grid} = \G{1}{grid}$, and that therefore $\str G_1$ satisfies $\Gamma_2$.

We are going to show that $\klasa_{\Gamma_2}$-FINSAT is undecidable, but 
$\klasa_{\Gamma_2}$-SAT is decidable, what gives us a proof of Theorem \ref{thm2}.

\subsection{The undecidability}

\begin{figure*}
\centering
\usetikzlibrary{calc,matrix,arrows}
\tikzstyle{vertex}=[circle,draw=black,minimum size=20pt,inner sep=0pt]
\tikzstyle{dots}=[minimum size=15pt,inner sep=0pt]
\tikzstyle{dvertex}=[circle,draw=black,minimum size=15pt, inner sep=0pt]

\begin{tikzpicture}[shorten >=0.5pt,->, scale=2.7]

\def \last {2.7}
\foreach \xx/\xxm/\xxxm/\xn in {0/0/1/0,1/1/2/1,2/2/0/2,4/1/2/4}
\foreach \yy/\yym/\yyym/\yn in {0/0/1/0,1/1/2/1,\last/1/2/3}
{
	\def \pos {p\xn b\yn}

 	\foreach \var/\x/\y/\name in 
	{
		P_{\xxm\yym}/0/0/1\pos,P_{\xxm\yyym}/0/1/2\pos,
		P_{\xxxm\yym}/1/0/3\pos,P_{\xxxm\yyym}/1/1/4\pos,
		A_{\xxm\yym}/0.27/0.27/5\pos,B_{\xxm\yym}/0.73/0.73/6\pos}
  	\node[vertex] (G-\name) at (\x + \xx,\y +\yy) {$\var$};

	 \foreach \var/\x/\y/\name in 
	 {
		F_{\xxm\yym}/0.5/0.5/10\pos}
	\node[dvertex] (G-\name) at (\x + \xx,\y+\yy) {${}_{\var}$};

	\foreach \from/\to in {
		1\pos/2\pos,1\pos/3\pos,2\pos/4\pos,3\pos/4\pos,
		1\pos/5\pos,6\pos/4\pos,5\pos/10\pos,10\pos/6\pos}
	\draw[arrows=-angle 90] (G-\from) -> (G-\to);

	\foreach \from/\to in {5\pos/4\pos}
	\draw[arrows={-angle 90}]  (G-\from) edge [in=250,out=10] (G-\to);
}

\foreach \yy in {0,1,3}
	\draw[arrows={-angle 90}, loosely dashed]  (G-4p4b\yy) edge [in=10,out=170] (G-2p0b\yy);

\foreach \yy in {0,3}
	\draw[arrows={-angle 90}, loosely dashed]  (G-3p4b\yy) edge [in=350,out=190] (G-1p0b\yy);

\foreach \yy in {0,4} 
	\draw[arrows={-angle 90}, loosely dashed]  (G-2p\yy b3) edge [in=101,out=259] (G-1p\yy b0);

\foreach \yy in {0,1,2,4} 
	\draw[arrows={-angle 90}, loosely dashed]  (G-4p\yy b3) edge [in=79,out=281] (G-3p\yy b0);

\coordinate [label=center:$\vdots$] (A) at (1.5,2.4);
\coordinate [label=center:$\cdots$] (A) at (3.5, 1);
\coordinate [label=center:$\adots$] (A) at (3.5, 2.4);
\end{tikzpicture}
\caption{The finite structure ${\str G}_2^m$. Edges from worlds satisfying $B_{ij}$ and worlds satisfying $E_{ij}^k$ and $G_{ij}^l$ are omitted for better readability. }\label{f:finite-model}
\end{figure*}
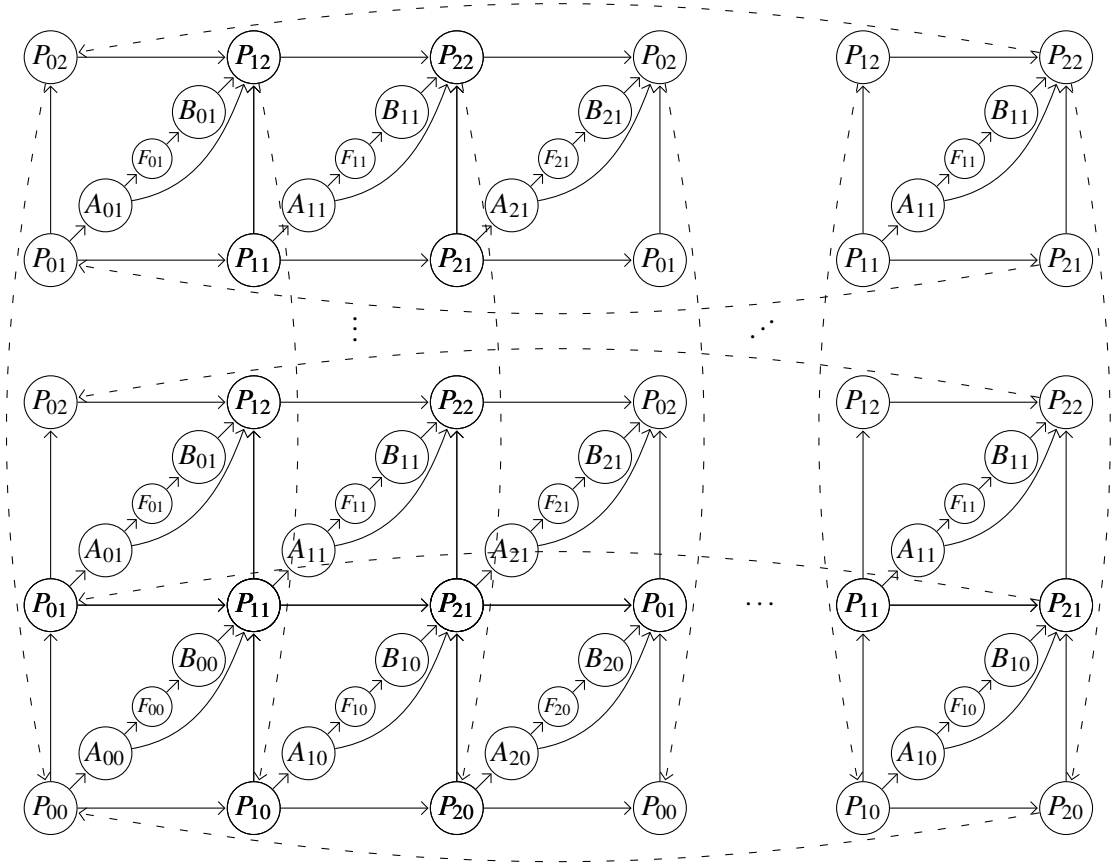

We now show that global $\klasa_{\Gamma_2}$-FINSAT is undecidable. 
In the proof we employ the second domino problem from Theorem \ref{doomino}.
Let ${\cal D} = (D,D_H,D_V)$ be a domino system and $\lambda^{{\cal D}}$ be
the modal encoding of ${\cal D}$ constructed in Section \ref{stare}. 
We define a modal formula $\tau_{2}$ such that $\tau_2 \wedge \lambda^{{\cal D}}$ has a finite $\klasa_{\Gamma_2}$ based model
if and only if ${\cal D}$ tiles $\Z_m \times \Z_m$, for some $m\in \N$. 
Actually, we simply put $\tau_2=\tau_1$.

\begin{lemma}
\label{lem2und}
Let ${\cal D} = (D,D_H,D_V)$ be a domino system and let $\tau_2 \wedge \lambda^{{\cal D}}$
be the formula constructed above. The tiling problem for ${\cal D}$ has a solution
if and only if $\tau_2 \wedge \lambda^{{\cal D}}$ has a finite $\klasa_{\Gamma_2}$ based model.    
\end{lemma}  
\begin{proof}
$\Rightarrow$
Let $t:\Z_m\times \Z_m \rightarrow D$  be a tiling of $\Z_m\times \Z_m$
for some $m\in \N$ and let ${\str G}^m_s$ be a standard grid structure over 
$\{w_{ij}\mid i,j\in [0,m-1]\}$ that represents $t$. Specifically, for all 
$i,j\in [0,m-1]$ a world $w_{ij}$ in ${\str G}^m_s$  has tile $t(i,j)$, adjacent tiles 
$(w_{ij},w_{(i+1 \mod m)j})$,$(w_{ij},w_{i(j+1 \mod m)})$ 
are connected by an R-edge, and pairs of horizontal (vertical) adjacent tiles respect
the relation $D_H$ (respectively $D_V$).

We expand ${\str G}^m_s$ to a finite model ${\str G}_2^m \in \klasa_{\Gamma_2}$ of $\tau_2 \wedge \lambda^{{\cal D}}$
by adding to ${\str G}^m_s$ some worlds and edges. We label each world $w_{ij}$
by $P_{(i \mod 3)(j \mod 3)}$ and for each such a world we add to ${\str G}_2^m$ 
worlds $a_{ij}, b_{ij}, f_{ij},g_{ij}^l, e_{ij}^k$ for $i,j\in [0, m-1]$, $k\in [0,2]$, $l \in [0,5]$. Label these worlds:
$a_{ij}$ by predicate $A_{(i \mod 3)(j \mod 3)}$, $b_{ij}$ by $B_{(i \mod 3)(j \mod 3)}$, 
$f_{ij}$ by $F_{(i \mod 3)(j \mod 3)}$ etc.
We then connect appropriate worlds of ${\str G}_2^m$ by R edges, as prescribed by the 
formula $\tau_2$ and connect each $b_{ij}$ to each world of ${\str G}_2^m$ (including itself).
Observe that all ramified worlds in a structure ${\str G}_2^m$ are $b_{ij}$ for $i,j\in [0,m-1]$. Thus both formula $\tau_2 \wedge \lambda^{{\cal D}}$
and $\G{2}{trans}$ are satisfied.

In order to see that $\G{2}{refram} \Rightarrow \G{2}{grid}$, we will show that $\G{2}{grid}$ is satisfied.
The only worlds $w'\in {\str G}_2^m$ satisfying $deg_{\geq 4}(w')$ are $b_{ij}$ and $a_{ij}$. Only these worlds can 
be substituted for $u$ in $\G{2}{grid}$ to satisfy the condition of the implication. 
Assume we have substituted a $b_{ij}$ for $u$. 
The only predecessors of $b_{ij}$ are $f_{ij}$ and 
$b_{i'j'}$ for all $i',j'\in [0,m-1]$, but only $b_{i'j'}$ have degree at least
two, thus only $b_{i'j'}$ should be considered as possible substitutions for $x$ 
in $\G{2}{grid}$. Since $b_{i'j'}$ is connected to all other worlds we then have
$b_{i'j'} R u$ where $u$ is any world of ${\str G}_2^m$. Thus the implication $\G{2}{grid}$ is satisfied for
such substitutions. Similarly, when we substitute an $a_{ij}$ for $u$ we notice that the only
predecessors of $a_{ij}$ that have degree not less than two are $b_{i'j'}$ and $w_{(i \mod m)(j \mod m)}$.
So only $b_{i'j'}$ and $w_{(i \mod m)(j \mod m)}$ should be considered as possible substitutions for $x$ 
in $\G{2}{grid}$. The case of $b_{i'j'}$ we have considered before and the case of 
$w_{(i \mod m)(j \mod m)}$ just says that ${\str G}_2^m$ should locally look like a grid, what is already enforced
by grid ${\str G}_s^m$. Therefore $\G{2}{grid}$ is satisfied and we have found a finite model of
$\tau_2 \wedge \lambda^{{\cal D}}$ and $\Gamma_2$.  

\medskip
$\Leftarrow$ 
We prove that every finite model $\str M$ of $\Gamma_2$ and $\tau_2 \wedge \lambda^{\cal D}$ 
contains a model of $\Gamma$ and $\tau \wedge \lambda^{\cal D}$ as a substructure.
Existence of such a model implies existence of the required tiling.   
Since $\str M$ satisfies $\tau_2 \wedge \lambda^{\cal D}$, each world of $\str M$
labeled by $P_{ij}, E_{ij}^k$ for $i,j,k\in [0,2]$ satisfies
$\tau \wedge \lambda^{\cal D}$. Notice that $A_{ij}$ fails to satisfy $\tau \wedge \lambda^{\cal D}$
only because it has a neighbor labeled $F_{ij}$. Let $\str N$ be a substructure of $\str M$ 
obtained by removing words labeled by $B_{ij}, G_{ij}^l, F_{ij}$ for $i,j\in [0,2], l \in [0,5]$
from $\str M$. Thus $\str N$ globally satisfies $\tau \wedge \lambda^{\cal D}$. We now show
that is also satisfies $\Gamma$. To achieve this observe that $\str M$ models $\G{2}{grid}$.
It is because there exists a world in $\str M$ that satisfies $\G{2}{refram}$.
Precisely, let $\str M_B$ be a substructure of $\str M$ induced by worlds of $\str M$ 
satisfying some $B_{ij}$, where $i,j\in [0,2]$. Note that $\str M_B$ is nonempty.
In the structure $\str M$, for each world satisfying $B_{ij}$ one can reach a world 
satisfying some $B_{i'j'}$ following a path  of $R$ edges. Moreover each world satisfying
$B_{ij}$ has at least $7$ successors in $\str M$. Then, because of
$\G{2}{trans}$ there is an infinite path in the finite structure $\str M_B$. Thus 
some world $b$ satisfying $B_{ij}$ is reflexive and the world $b$ satisfies $\G{2}{refram}$.
Therefore $\str M$ satisfies $\G{2}{grid}$.       

Now notice that if  $\str M \models w_x R w_u \wedge  w_u R w_z \wedge w_x R w_z$
for some worlds $w_x,w_y,w_z \in M$ then at least one of these worlds is labeled
by $B_{ij}$ for some $i,j\in [0,2]$. Therefore the structure $\str N$ satisfies:
\begin{equation}
 x R u \wedge u R z \Rightarrow \neg x R z
\label{eq:anti-trans-second}
\end{equation}
Clearly, $\G{2}{grid}$ and \eqref{eq:anti-trans-second} imply $\Gamma$. Thus, $\str N$ satisfies $\Gamma$.
Hence, we have found a finite model $\str N \in \klasa_{\Gamma}$ of $\tau \wedge \lambda^{\cal D}$.

\end{proof}
Therefore we have the following
\begin{corollary}
$\klasa_{\Gamma_2}$-FINSAT is undecidable.
\end{corollary}

\subsection{The decidability}

We say that a frame $\cal M$ is a \emph{quasi-tree} if there is a tree $\cal T$ over the same universe such that there is an edge from $w$ to $v$ in $\cal M$ if and only if there is such edge in $\cal T$ or $w$ is ramified and there is a path from $w$ to $v$ in $\cal T$. In other worlds, a quasi-tree looks like a tree except that ramified worlds are connected to all their descendants, not only successors.  We are ready to prove the following lemma.
\begin{lemma}
\label{thm2dec}
$\klasa_{\Gamma_2}$-SAT is decidable.
\end{lemma}
\begin{proof}
First, we show that modal logic over $\klasa_{\Gamma_2}$ has a quasi-tree model property, i.e. each satisfiable formula has a model based on a quasi-tree. Then, we describe an \APSPACE{} algorithm that solves $\klasa_{\Gamma_2}$-SAT by checking if a given formula has a quasi-tree model.

Let $\str M$ be a model of $\varphi$ and $\Gamma_2$ and $\str M'$ be result of its unraveling (see \cite{BdRV01}). Clearly, $\str M'$ satisfies $\varphi$ but not necessarily $\Gamma_2$. We define $\str N$ as a model that contains $\str M'$ and additional edges --- we connect all ramified worlds with all their descendants. A quick check shows that $\str N$ is a model of $\varphi$ and $\Gamma_2$ based on a quasi-tree.

Now we describe an alternating algorithm that verifies existence of a model based on a quasi-tree. It guesses a model starting from root, and keeps the information about the subformulas satisfied by ramified worlds along the path. 
First, it guesses a type of an initial world.  Then recursively, it guesses at most $|\varphi|$ types of successors of a current world, verifies that they are consistent with all subformulas of the form $\square \psi$ satisfied by a current world and all ramified worlds above, and to guarantee that all subformulas of the form $\Diamond \psi$ are satisfied in the current world. Then, it calls itself universally for each successor.  After $2^{|\varphi|}+1$ steps, we know that the algorithm visited some type twice. So we keep a counter of steps and, when it reaches $2^{|\varphi|}+1$, the algorithm accepts.

Clearly, the described algorithm needs only polynomial space, proving the membership of $\klasa_{\Gamma_2}$-SAT in \APSPACE{}.
We claim that this bound is tight, but we skip the lower bound proof.
\end{proof}

\section{Conclusion and future work}
We showed that the decidability of the global satisfiability problem of elementary modal logics may vary on the decision whether we consider only finite structures or not. 

Of course, $\Gamma_1$ and $\Gamma_2$ are not that only formulas with the desired properties. We can easily show that there are infinitely many such formulas.

\begin{proposition}
There is infinitely many non-equivalent universal first-order formulas $\Phi$ such that $\klasap$-SAT is decidable and $\klasap$-FINSAT is undecidable, and infinitely may non-equivalent universal first-order formulas $\Phi$ such that $\klasap$-SAT is undecidable and $\klasap$-FINSAT is decidable.
\end{proposition}

Consider a formula $\Lambda_n = \forall x_1 \dots x_n. \bigvee_{i \neq j} (\neg x_i R x_j \vee x_i = x_j \vee x_i R x_i)$ stating  that a structure do not contain any (irreflexive) clique of size $n$. It is not hard to see that all the models considered in this paper satisfy $\Lambda_n$ for any $n>2$. Therefore, for any $n>2$, the formula $\Gamma_1 \wedge \Lambda_n$ defines an undecidable logic that is finitely decidable, and the formula $\Gamma_2 \wedge \Lambda_n$ defines decidable logic that is finitely undecidable. Therefore, in both cases, the number of such formulas is infinite.

There are two natural questions that concern local satisfiability.\\ \indent 
\emph{ Question 1}.  Is there a universal first-order formula that defines an elementary modal logic with undecidable local satisfiability problem and decidable finite local satisfiability problem?\\
\indent\emph{ Question 2}.  Is there a universal first-order formula that defines an elementary modal logic with decidable local satisfiability problem and undecidable finite local satisfiability problem?

Equality plays a crucial role in our proofs. We do not know, however, whether it is necessarily. \\ \indent  \emph{ Question 3}. Can our results be proved without using equality?

We have seen that assuming finiteness can change the decidability of the logic. Another interesting question would be ``How much can the complexity differ, if both problems are decidable?'' However, there is even a simpler question that we would like to state.\\ \indent \emph{ Question 4}. Is there a universal first-order formula $\Phi$ such that the global satisfiability problem of modal logic over $\klasap$ is \NEXPTIME{}-hard and decidable?  How about local satisfiability, finite global satisfiability and finite local satisfiability problems?
\bibliographystyle{eptcs}
\bibliography{all}
\end{document}